\documentclass[twoside]{article}

\usepackage[accepted]{aistats2023}

\usepackage{graphicx}
\usepackage{subcaption}
\usepackage{todonotes}
\usepackage{bm}
\usepackage{amsmath}
\usepackage{amsthm}
\usepackage{todonotes}

\usepackage{pxfonts}
 
\usepackage[round]{natbib}

\newtheorem{theorem}{Theorem}
\newtheorem{definition}{Definition}
\newtheorem{lemma}{Lemma}

\def\ci{\perp\!\!\!\perp}

\begin{document}
\runningtitle{Combining Graphical and Algebraic Approaches for Parameter Identification in Latent Variable SEM}
\twocolumn[
	\aistatstitle{Combining Graphical and Algebraic Approaches for \\ Parameter Identification in Latent Variable Structural Equation Models}
\aistatsauthor{ Ankur Ankan \And Inge Wortel \And Kenneth A. Bollen \And Johannes Textor}
\aistatsaddress{ Radboud University \And Radboud University \And University of North Carolina \\ at Chapel Hill \And Radboud University} ]

%

\begin{abstract}
	Measurement error is ubiquitous in many variables -- from blood
	pressure recordings in physiology to intelligence measures in
	psychology. Structural equation models (SEMs) account for the process
	of measurement by explicitly distinguishing between \emph{latent}
	variables and their measurement \emph{indicators}. Users often fit
	entire SEMs to data, but this can fail if some model parameters are not
	identified. The model-implied instrumental variables (MIIVs) approach
	is a more flexible alternative that can estimate subsets of model
	parameters in identified equations. Numerous methods to identify
	individual parameters also exist in the field of graphical models (such
	as DAGs), but many of these do not account for measurement effects.
	Here, we take the concept of ``latent-to-observed'' (L2O)
	transformation from the MIIV approach and develop an equivalent
	graphical L2O transformation that allows applying existing graphical
	criteria to latent parameters in SEMs. We combine L2O transformation
	with graphical instrumental variable criteria to obtain an efficient
	algorithm for non-iterative parameter identification in SEMs with
	latent variables. We prove that this graphical L2O transformation with
	the instrumental set criterion is equivalent to the state-of-the-art
	MIIV approach for SEMs, and show that it can lead to novel
	identification strategies when combined with other graphical criteria. 
\end{abstract}

\section{\MakeUppercase{Introduction}}

Graphical models such as directed acyclic graphs (DAGs) are currently used in
many disciplines for causal inference from observational studies. However, the
variables on the causal pathways modelled are often different from those being
measured. Imperfect measures cover a broad range of sciences, including health
and medicine (e.g., blood pressure, oxygen level), environmental sciences
(e.g., measures of pollution exposure of individuals), and the social (e.g.,
measures of socioeconomic status) and behavioral sciences (e.g., substance
abuse).

Many DAG models do not differentiate between the variables on the causal
pathways and their actual measurements in a dataset \citep{Tennant2019}.  While
this omission is defensible when the causal variables can be measured reliably
(e.g., age), it becomes problematic when the link
between a variable and its measurement is more
complex.  For example, graphical models employed in fields like Psychology or
Education Research often take the form of \emph{latent variable structural
equation models} (LVSEMs, Figure~\ref{fig:example_sem}; \citet{bollen1989}),
which combine a \emph{latent level} of unobserved variables and their
hypothesized causal links with a \emph{measurement level} of their observed
indicators (e.g., responses to questionnaire items). This structure is so
common that LVSEMs are sometimes simply referred to as SEMs. In contrast,
models that do not differentiate between causal factors and their measurements
are traditionally called \emph{simultaneous equations} or \emph{path
models}\footnote{Path models can be viewed as
LVSEMs with all noise set to 0; some work on path models,
importantly by Sewall Wright himself, does incorporate latent variables. }.

Once a model has been specified,
estimation can be performed in different ways. SEM parameters are
often estimated all at once by iteratively minimizing some difference measure
between the observed and the model-implied covariance matrices. However, 
this ``global'' approach has some pitfalls. First, all model parameters must be
algebraically identifiable for a unique minimum to exist; if only a single
model parameter is not identifiable, the entire fitting procedure may
not converge \citep{boomsma1985nonconvergence} or provide meaningless results.
Second, local model specification errors can propagate through the entire model
\citep{bollen2007latent}. Alternatively, \citet{bollen1996alternative}
introduced a ``local'', equation-wise approach for SEM parameter identification termed
``model-implied instrumental variables'' (MIIVs), which is non-iterative and
applicable even to models where not all parameters are simultaneously
identifiable. MIIV-based SEM identification is a mature approach with a
well-developed underlying theory as well as implementations in multiple
languages, including R \citep{fisher2019miivsem}.  

\begin{figure}
	\centering
	\includegraphics[page=1]{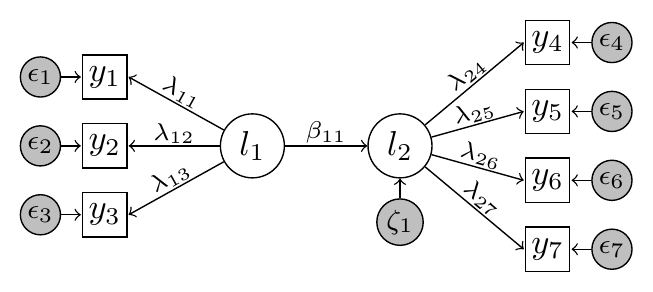}
	\caption{SEM based on the \emph{Industrialization and
		 Political Democracy} model \citep{bollen1989} with latent
		 variables $ l_1 $ (industrialization), and $ l_2 $ (political democracy).
		 The model contains 3 indicators for $ l_1 $: (1) gross
		 national product ($ y_1 $), (2) energy consumption
		 ($ y_2 $), and (3) labor force in industry ($ y_3 $), and 4
		 indicators for $ l_2 $: (1) press freedom rating ($y_4$), (2)
		 political opposition freedom ($y_5$), (3) election fairness
		 ($y_6$), and (3) legislature effectiveness ($y_7$). $
		 \lambda_{11} \dots \lambda_{13}, \lambda_{24} \dots
		 \lambda_{27}, \text{ and } \beta_{11} $ are the path
		 coefficients. $ \epsilon_1, \dots, \epsilon_7, \text{ and }
		 \zeta_1 $ represent noise/errors.}
	\label{fig:example_sem}
\end{figure}

Of all the model parameters that are identifiable in principle, any given
estimator (such as the MIIV-based approach) can typically only identify
parameters in identified equations and identified parameters in underidentified
equations. Different identification methods are therefore complementary and can
allow more model parameters to be estimated. Having a choice of such
methods can help
users to keep the stages of \emph{specification} and
\emph{estimation} separated. 
For example, a researcher who only has access to global
identification methodology might be tempted to impose model restrictions just
to ``get a model identified'' and not because there is a theoretical rationale
for the restrictions imposed. 
With more complementary methods to choose from, researchers can 
instead base model specification on substantive theory and causal assumptions.

The development of parameter identification methodology has received intense
attention in the graphical modeling field. The most general identification
algorithm is Pearl's do-calculus, which provides a complete solution in
non-parametric models
\citep{DBLP:conf/uai/HuangV06,DBLP:conf/aaai/ShpitserP06}.  The back-door and
front-door criteria provide more convenient solutions in special cases
\citep{pearl2009causality}.  While there is no practical general algorithm to
decide identifiability for models that are linear in their parameters, there has been a flurry of
work on graphical criteria for this case, such as instrumental sets
\citep{BritoP02}, the half-trek criterion \citep{Foygel2012}, and auxiliary
variables \citep{ChenKB17}. Unfortunately, these methods were all developed for
the acyclic directed mixed graph (ADMG) framework and require at least the
variables connected to the target parameter to be observed -- which is rarely
the case in SEMs. Likewise, many criteria in graphical models are based on
``separating'' certain paths by conditioning on variables, whereas no such
conditioning-based criteria exist for SEMs.

The present paper aims to make identification
methods from the graphical model literature available to the SEM field.
We offer the following contributions:
\begin{itemize}
\item We note that \citet{bollen1996alternative}'s latent-to-observed (L2O)
	transformation that transforms a latent variable SEM into a model with
	only observed variables can be used more generally in models containing 
	arbitrary mixtures of latent and observed variables
	(Section~\ref{sec:l2o}).
\item We present a graphical equivalent of L2O transformation that allows us to
	apply known graphical criteria to SEMs (Section~\ref{sec:graphical_l2o}).
\item We prove that Bollen's MIIV approach \citep{bollen1996alternative,bollen2004automating,Bollen2022} is
	equivalent to a graphical L2O transformation followed by the application of
	the graphical instrumental set criterion (\citet{BritoP02}; Section~\ref{sec:equiv}).
\item We give examples where the graphical L2O transformation approach 
	can identify more parameters compared to the MIIV approach 
	implemented in the R package MIIVsem (\citet{fisher2019miivsem}; Section~\ref{sec:examples}).
\end{itemize}
Thus, by combining the L2O transformation idea from the SEM literature with
identification criteria from the graphical models field, we bridge
these two fields -- hopefully to the benefit of both.

\section{\MakeUppercase{Background}}
\label{sec:background}

In this section, we give a brief background on basic graphical terminology
and define SEMs.

\subsection{Basic Terminology}

We denote variables using lowercase letters ($x_i$), sets and vectors of
variables  using uppercase letters ($X$), and matrices using boldface
($\bm{\Lambda}$).  We write the cardinality of a set $V$ as $|V|$, and the rank of
a matrix $ \bm{\Lambda} $ as $ \textrm{rk}(\bm{\Lambda}) $.  A \emph{mixed
graph} (or simply \emph{graph}) ${\cal G}=(V,A)$ is defined by sets of
variables (nodes) $V=\{x_1,\ldots,x_n\}$ and arrows $A$, where arrows can be
directed ($x_i \to x_j$) or bi-directed $(x_i \leftrightarrow x_j)$.  A
variable $x_i$ is called a \emph{parent} of another variable $x_j$ if $x_i \to
x_j \in A$, or a \emph{spouse} of $x_j$ if  $x_i \leftrightarrow x_j \in A$. We
denote the set of parents of $ x_i $ in $ {\cal G} $ as $ Pa_{\cal G}(x_i) $.

{\bfseries Paths:} A \emph{path} of length $k$ is a sequence of $k$ variables
such that each variable is connected to its neighbours by an arrow.  A
\emph{directed path} from $x_i$ to $x_j$ is a path on which all arrows point
away from the start node $x_i$. For a path $\pi$, let $\pi[x_i \sim x_j]$
denote its subsequence from $x_i$ to $x_j$, in reverse order when $x_i$ occurs
after $x_j$; for example, if $\pi=x_1 \gets x_2 \to x_3$ then $\pi[x_2 \sim
x_3]=x_2 \to x_3$ and $\pi[x_1 \sim x_2]=x_2 \to x_1$.  
Importantly, this definition of a path is common in DAG literature
but is different from the SEM literature, where ``path'' typically
refers to a single arrow between two
variables. Hence, a path in a DAG is equivalent to a sequence of paths in path
models. An \emph{acyclic directed mixed graph} (ADMG) is a mixed graph with no
directed path of length $\geq 2$ from a node to itself.

{\bfseries Treks and Trek Sides:} A \emph{trek} (also called \emph{open path})
is a path that does not contain a \emph{collider}, that is, a subsequence $ x_i
\to x_j \gets x_k $. A path that is not
open is a \emph{closed path}. Let $\pi$ be a trek from $x_i$ to $x_j$, then
$\pi$ contains a unique variable $t$ called the \emph{top}, also written as
$\pi^\leftrightarrow$, such that $\pi[t \sim x_i]$ and $\pi[t \sim x_j]$ are
both directed paths (which could both consist of a single node).
Then we call $\pi^\gets := \pi[t \sim x_i]$ the \emph{left side} and $\pi^\to
:= \pi[t \sim x_j]$ the \emph{right side} of $\pi$.\footnote{In the literature, treks are also often represented
as tuples of their left and right sides.}

{\bfseries Trek Intersection:} Consider two treks $\pi_i$ and $\pi_j$, then we
say that $\pi_i$ and $\pi_j$ \emph{intersect} if they contain a common variable
$v$. We say that they \emph{intersect on the same side} (have a same-sided
intersection) if $v$ occurs on $\pi_i^\gets$ and $\pi_j^\gets$ or $\pi_i^\to$
and $\pi_j^\to$; in particular, if $v$ is the top of $\pi_i$ or $\pi_j$, then
the intersection is always same sided.  Otherwise, $\pi_i$ and $\pi_j$
\emph{intersect on opposite sides} (have an opposite-sided intersection).

{\bfseries t-separation:} Consider two sets of variables, $L$ and $R$, and a set
$T$ of treks. Then we say that the tuple $(L,R)$ \emph{$t$-separates} (is a
$t$-separator of) $T$ if every trek in $T$ contains either a variable in $L$ on
its left side or a variable in $R$ on the right side. For two sets of variables,
$A$ and $B$, we say that $(L,R)$ $t$-separates  $A$ and $B$ if it $t$-separates
all treks between $A$ and $B$. The \emph{size} of a $t$-separator $(L,R)$ is
$|L|+|R|$.

\subsection{Structural Equation Models}
\label{subsec:sem_ram}

We now define structural equation models (SEMs) as they are usually considered
in the DAG literature \citep[e.g.,][]{sullivant2010trek}. This
definition is the same as the Reticular Action Model (RAM) representation
\citep{mcardle1984some} from the SEM literature. A \emph{structural equation
model} (SEM) is a system of equations linear in their parameters such that:
\[
X = \bm{B} X + E
\]
where $ X $ is a vector of variables (both latent and observed), $\bm{B}$ is a
$ |X| \times |X| $ matrix of \emph{path coefficients}, and
$E=\{\epsilon_1,\ldots,\epsilon_{|X|}\}$ is a vector of error
terms with a positive definite covariance matrix $\bm{\Phi}$ (which has typically
many or most of its off-diagonal elements set to 0) and zero
means.\footnote{This can be extended to allow for non-zero means, but our focus
here is on the covariance structure, so we omit that for simplicity.} The
\emph{path diagram} of an SEM $(\bm{B},\bm{\Phi})$ is a mixed graph with nodes
$V=X \cup E $ and arrows $A = \{ \epsilon_i \to x_i \mid i \in 1,\ldots, |X| \}
\cup \{ x_i \to x_j \mid \bm{B}[i,j] \neq 0 \} \cup \{ \epsilon_i
\leftrightarrow \epsilon_j \mid i \neq j, \bm{\Phi}[i,j] \neq 0 \} $.  We also
write $\beta_{x_i \to x_j}$ for the path coefficients in $\bm{B}$ and
$\phi_{\epsilon_i}$ for the diagonal entries (variances) in $\bm{\Phi}$. 
Each equation in the model corresponds to one node in this graph, where
the node is the dependent variable and its parent(s) are the explanatory 
variable(s). Each arrow represents one parameter to be
estimated, i.e., a \emph{path coefficient} (e.g., directed arrow between latents and
observed variables), a \emph{residual covariance} (bi-directed arrow), or a
\emph{residual variance} (directed arrow from error term to latent or
indicator). However, some of these parameters could be fixed; for example,
at least one parameter per latent variable needs to be fixed to set its scale,
and covariances between observed exogenous variables (i.e., observed
variables that have no parents) are typically fixed to their observed values.
In this paper, we focus on estimating the path coefficients.  We
only consider \emph{recursive} SEMs in this paper -- i.e., where the path
diagram is an ADMG -- even though the methodology can be generalized.

%

\citet{sullivant2010trek} established an
important connection between treks and the ranks of submatrices of the 
covariance matrix, which we will heavily rely on in our paper. 
\begin{theorem}{Trek separation; \citep[Theorem~2.8]{sullivant2010trek}}
\label{thm:treksep}
Given an SEM $\cal G$ with an implied covariance matrix $ \bm{\Sigma} $, and two subsets of variables $A,B \subseteq X$,
\[
\text{rk}(\bm{\Sigma}[A,B]) \leq \text{min}\left\{|L|+|R| \mid (L,R) \text{
 t-separates } A \text{ and } B \right\}
\] 
where the inequality is tight for generic covariance matrices implied by 
$\cal G$.
\end{theorem}

In the special case $A=\{x_1\},B=\{x_2\}$, Theorem~\ref{thm:treksep} implies
that $x_1$ and $x_2$ can only be correlated if they are connected by a trek.
Although the compatible covariance matrices of SEMs can also be characterized
in terms of $d$-separation \citep{chen2014graphical}, we use $t$-separation for
our purpose because it does not require conditioning on variables, and it
identifies more constraints on the covariance matrix implied by SEMs
than d-separation \citep{sullivant2010trek}.

\section{\MakeUppercase{Latent-To-Observed Transformations for SEMs}}
\label{sec:l2o}

A problem with IV-based identification criteria is that they cannot be directly
applied to latent variable parameters. The MIIV approach addresses this issue
by applying the L2O transformation to these model equations, such that they
only consist of observed variables. The L2O transformation in
\citet{bollen1996alternative} is presented on the LISREL representation of SEMs
(see Supplementary Material). In this section, we first briefly introduce
``scaling indicators'', which are required for performing L2O transformations.
We then use it to define the L2O transformation on the RAM notation (defined in
Section~\ref{subsec:sem_ram}) and show that with slight modification to the
transformation, we can also use it to partially identify equations. We will, from 
here on, refer to this transformation as the ``algebraic L2O
transformation'' to distinguish it from the purely graphical L2O transformation
that we introduce later in Section~\ref{sec:graphical_l2o}.


\subsection{Scaling Indicators}
The L2O transformation (both algebraic and graphical) uses the fact that any
SEM is only identifiable if the scale of each latent variable is fixed to an
arbitrary value (e.g., 1), introducing new algebraic constraints.  These
constraints can be exploited to rearrange the model equations in such a
way that latent variables can be eliminated.

The need for scale setting is well known in the SEM literature and arises from
the following lemma (since we could not find a direct proof in the literature
-- perhaps due to its simplicity -- we give one in the Appendix).
\begin{lemma} (Rescaling of latent variables). 
	\label{lemma:scaling}
	Let $x_i$ be a variable in an SEM $(\bm{B},\bm{\Phi})$.  Consider another SEM
	$(\bm{B}',\bm{\Phi}')$ where we choose a scaling factor $\alpha \neq 0$
	and change the coefficients as follows: For every parent $p$ of $x_i$,
	$\beta_{p \to x_i}' = \alpha^{-1} \; \beta_{p \to x_i}$; for every
	child $c$ of $x_i$, $\beta_{x_i \to c}' = \alpha \; \beta_{x_i \to c}$;
	for every spouse $s$ of $x_i$, $\phi_{x_i \leftrightarrow s}' =
	\alpha^{-1} \phi_{x_i \leftrightarrow s}$; and $\phi_{x_i}'=
	\alpha^{-2} \phi_{x_i}$.  Then for all $j,k \neq i$,
	$\bm{\Sigma}[j,k]=\bm{\Sigma}'[j,k]$.
\end{lemma}

If $x_i$ is a latent variable in an SEM, then Lemma~\ref{lemma:scaling} implies
that we will get the same implied covariance matrix among the observed
variables for all possible scaling factors. In other words, we need to set the
scale of $x_i$ to an arbitrary value to identify any parameters in such a
model.  Common choices are to either fix the error variance of every latent
variable such that its total variance is $ 1 $, or to choose one indicator per
latent and set its path coefficient to $ 1 $. The latter method is often
preferred because it is simpler to implement. The chosen indicators for each
latent are then called the \emph{scaling indicators}. However, note that
Lemma~\ref{lemma:scaling} tells us that we can convert any fit based on scaling
indicators to a fit based on unit latent variance, so this choice does not
restrict us in any way.

\subsection{Algebraic L2O Transformation for RAM}
The main idea behind algebraic L2O transformation is to replace each of the
latent variables in the model equations by an observed expression involving
the scaling indicator. 
As in \citet{bollen1996alternative}, we assume that each of the latent
variables in the model has a unique scaling indicator. We show the
transformation on a single model equation to simplify the notation. Given an
SEM $\cal G$ on variables $ X $, we can write the equation of any variable $
x_i \in X $ as:
\[
x_i = \epsilon_i + \sum_{x_j \in Co(x_i)}\beta_{x_j \rightarrow x_i} x_j
\]

where $ Co(x_i) = \{ Co_l(x_i), Co_o(x_i) \} $ is the set of covariates in the
equation for $ x_i $. $ Co_l(x_i) $ and $ Co_o(x_i) $ are the latent and
observed covariates, respectively. Since each latent variable  $x_j$ has a
unique scaling indicator $x_j^s$, we can write the latent variable as
 $ x_j = x_j^s - \epsilon_{x_j^s} $.
Replacing all the latents in the above equation with their scaling indicators,
we obtain:
\[
 		x_i = \epsilon_i + \sum_{x_j \in Co_l(x_i)} \beta_{x_j \rightarrow x_i} (x_j^s - \epsilon_{x_j^s}) + \sum_{x_k \in Co_o(x_i)} \beta_{x_k \rightarrow x_i} x_k
\]


If $ x_i $ is an observed variable, the transformation is complete as the
equation only contains observed variables. But if $ x_i $ is a latent variable,
we can further replace $ x_i $ as follows:

\[
	x_i^s = \epsilon_i + \epsilon_{x_i^s} + \sum_{x_j \in Co_l(x_i)} \beta_{x_j \rightarrow x_i} (x_j^s - \epsilon_{x_j^s}) + \sum_{x_k \in Co_o(x_i)} \beta_{x_k \rightarrow x_i} x_k
\]

As the transformed equation now only consists of observed variables, IV-based
criteria can be applied to check for identifiability of parameters.

\subsection{Algebraic L2O Transformations for Partial Equation Identification}
\label{sec:partial_transform}
In the previous section, we used the L2O transformation to replace all the
latent variables in the equation with their scaling indicators, resulting in an
equation with only observed variables. An IV-based estimator applied to these
equations would try to estimate all the parameters together. However, there are
cases (as shown in Section~\ref{sec:examples}) where not all of the parameters
of an equation are identifiable. If we apply L2O transformation to the whole
equation, none of the parameters can be estimated.

Here, we outline an alternative, ``partial'' L2O transformation that
replaces only some of the latent variables in the equation. Assuming $
Co_l^i(x_i) \subset Co_l(x_i) $ as the set of latent variables whose parameters
we are interested in estimating, we can write the partial L2O transformation
as:
\[
	\begin{split}
		x_i = \epsilon_i + & \sum_{x_j \in Co_l^i(x_i)} \beta_{x_j \rightarrow x_i} (x_j^s - \epsilon_{x_j^s}) +  \\
				   & \sum_{x_k \in Co_l(x_i) \setminus Co_l^i(x_i)} \beta_{x_k \rightarrow x_i} x_k + \sum_{x_l \in Co_o(x_i)} \beta_{x_l \rightarrow x_i} x_l
	\end{split}
\]

Similar to the previous section, we can further apply L2O transformation for $
x_i $ if it is also a latent variable. As the parameters of interest are now
with observed covariates in the transformed equation, IV-based criteria can be
applied to check for their identifiability while treating the variables in $
Co_l(x_i) \setminus Co_l^i(x_i) $ as part of the error term.

\section{\MakeUppercase{Graphical L2O Transformation}}
\label{sec:graphical_l2o}

Having shown the algebraic L2O transformation, we now show that these
transformations can also be done graphically for path diagrams. An
important difference is that the algebraic transformation
is applied to all equations in a model simultaneously by 
replacing all latent variables, whereas we 
apply the graphical transform only to a single equation at a time (i.e.,
starting from the original graph for every equation). Applying the graphical
transformation to multiple equations simultaneously results in a 
non-equivalent model with a different implied covariance matrix.

Given an SEM $\cal G$, the equation for any variable $x_j$ can be written in
terms of its parents in the path diagram as: $ x_j = \sum_{x_k \in
\textrm{Pa}_{\cal G}(x_j)} \beta_{x_k \to x_j} x_k + \epsilon_{x_j} $. Using this equation,
we can write the relationship between any latent variable $ x_j $ and its scaling 
indicator $ x_j^s $ as (where $ \beta_{x_j \to x_j^s} $ is fixed to $ 1 $):

\begin{equation}
	\label{eq:graphical_transform}
	x_j = x_j^s - \epsilon_{x_j^s} - \sum_{x_k \in \textrm{Pa}_{\cal G}(x_j^s) \setminus x_j} \beta_{x_k \to x_j^s} x_k
\end{equation}

We use this graphical L2O transformation as follows. Our goal is to identify a path
coefficient $\beta_{x_i \to x_j}$ in a model $\cal G$.  If both $x_i$ and $x_j$
are observed, we leave the equation untransformed and apply graphical 
identification criteria \citep{chen2014graphical}.  Otherwise,
we apply the graphical L2O transformation to $\cal G$ with respect to $ x_i $,
$ x_j$, or both  variables -- ensuring that the resulting model ${\cal G}'$
contains an arrow between two observed variables $x_i'$ and $x_j'$, where the
path coefficient $\beta_{x_i' \to x_j'}$ in ${\cal G}'$ equals $\beta_{x_i \to
x_j}$ in $\cal G$.

We now illustrate this approach on an example for each of the three possible
combinations of latent and observed variables.

\begin{figure}[t]
    \centering
    \begin{subfigure}[b]{.33\linewidth}
    	\centering
	\includegraphics[page=4]{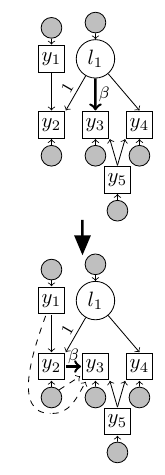}
        \caption{}
	\label{fig:l2o_parent}
    \end{subfigure}%
    \begin{subfigure}[b]{.33\linewidth}
    	\centering
	\includegraphics[page=5]{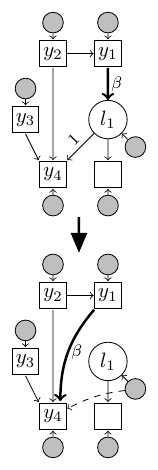}
        \caption{}
	\label{fig:l2o_child}
    \end{subfigure}%
    \begin{subfigure}[b]{.33\linewidth}
    	\centering
	\includegraphics[page=6]{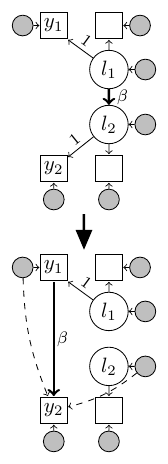}
        \caption{}
	\label{fig:l2o_both}
    \end{subfigure}
        \caption{Example L2O transformations for path coefficients
(a) from a latent to an observed variable; 
(b) from an observed to a latent variable; 
(c) between two latent variables.}
\end{figure}

{\bfseries Latent-to-observed arrow:} Consider the arrow $l_1 \to y_3 $ in
Figure~{\ref{fig:l2o_parent}}, and let $\beta$ be the path coefficient of this
arrow. To perform the L2O transformation, we start with the
model equation involving $\beta$: 
\[
y_3 = \beta l_1 + \beta_{y_5 \to y_3} y_5 + \epsilon_3
\]
We then use Equation~\ref{eq:graphical_transform} to write the latent variable,
$ l_1 $ in terms of its scaling indicator, $y_2$ as: $ l_1 = y_2 - \epsilon_{2} +
\beta_{y_1 \to y_2} y_1 $, and replace it in the above equation to obtain:
\[
	y_3 = \beta y_2 - \beta \beta_{y_1 \to y_2} y_1 + \beta_{y_5 \to y_3} y_5 - \beta \epsilon_2 + \epsilon_3
\]
The transformation has changed the equation for $ y_3 $, which now regresses on
the observed variables $ y_2 $, $ y_1 $, and $ y_5 $, as well as the errors $
\epsilon_2 $ and $ \epsilon_3 $. We make the same changes in the graphical
structure by adding the arrows $ y_2 \to y_3 $, $ y_1 \to y_3 $, $ \epsilon_2
\to y_3 $, and removing the arrow $ l_1 \to y_3 $.

{\bfseries Observed-to-latent arrow:} Consider the arrow $y_1 \to l_1$ in
Figure~\ref{fig:l2o_child} with coefficient $\beta$. For L2O transformation in
this case, we apply Equation~\ref{eq:graphical_transform} to replace $ l_1 $ in
the model equation $ l_1 = \beta y_1 + \zeta_1 $ to obtain:
\[
	y_4 = \beta y_1 + \beta_{y_3 \to y_4} y_3 + \beta_{y_2 \to y_4} y_2 + \zeta_1 + \epsilon_4
\]
The equivalent transformation to the path diagram consists of adding the arrows $ y_1 \to y_4
$, and $ \zeta_1 \to y_4 $, and removing the arrows: $ l_1 \to y_4 $ and $ y_1 \to l_1 $.

{\bfseries Latent-to-latent arrow:} Consider the arrow $ l_1 \to l_2$ in
Figure~\ref{fig:l2o_both} with coefficient $ \beta $. In this case, we again
apply Equation~\ref{eq:graphical_transform} to replace both $ l_1 $ and $ l_2 $ in the model 
equation for $ l_2 = \beta l_1 + \zeta_2 $. This is equivalent to applying two L2O transformations in sequence
and leads to the transformed equation:
\[
	y_2 = \beta y_1 - \beta \epsilon_1 + \zeta_2 + \epsilon_2
\]

Equivalently, we now add the arrows $ y_1 \to y_2 $, $
\zeta_2 \to y_2 $, and $ \epsilon_1 \to y_2 $. We also remove the arrows $
l_2 \to y_2 $ and $ l_1 \to l_2 $.

\section{\MakeUppercase{Model-Implied Instrumental Variables Are Equivalent to Instrumental Sets}}
\label{sec:equiv}

After applying the L2O transformations from the previous sections, we can
use either algebraic or graphical criteria to check whether the path
coefficients are identifiable. In this section, we introduce the Instrumental
set criterion \citep{BritoP02} and the MIIV approach from
\citet{bollen1996alternative} that precedes it, and show that they are
equivalent. Importantly, even though we refer to the MIIV
approach as an algebraic criterion to distinguish it from the graphical
criterion, it is not a purely algebraic approach and utilizes the graphical
structure of the model to infer correlations with error terms.

We will first focus on the instrumental set criterion proposed by 
\citet{BritoP02}. We state the criterion below in a slightly rephrased
form that is consistent with our notation in Section~\ref{sec:background}:

\begin{definition}[Instrumental Sets \citep{BritoP02}]
\label{defn:graphicis}
Given an ADMG $\cal
	G$, a variable $y$, and a subset $X$ of the parents of $y$, 
	a set of variables
	$I$ fulfills the 
	\emph{instrumental set condition}
	if for {some} permutation $ i_1 \ldots i_k $ of
	$ I $ and {some} permutation
	$ x_1 \ldots x_k $ of $ X $ we have: 
	\begin{enumerate}
		\item There are no treks from $I$ to $y$ in the graph ${\cal
			G}_{\overline{X}}$ obtained by removing all arrows 
			between $X$ and $y$. 
		\item For each $j$, $1 \leq j \leq k$, there is a trek $\pi_j$ from
			$I_j$ to $X_j$ such that for all $i < j$: (1) $I_i$ does not
			occur on any trek $\pi_j$; and (2) all intersections between
			$\pi_i$ and $\pi_j$ are on the left side of $\pi_i$ and the
			right side of $\pi_j$.
	\end{enumerate}
\end{definition}

Its reliance on permutation makes the instrumental set criterion fairly
complex; in particular, it is not obvious how an algorithm to find such sets
could be implemented, since enumerating all possible permutations and paths is
clearly not a practical option. Fortunately, we can rewrite this criterion into
a much simpler form that does not rely on permutations and has an obvious
algorithmic solution.

\begin{definition}[Permutation-free Instrumental Sets]
\label{defn:graphicistrek}
	Given an ADMG $\cal G$, a variable $y$ and a subset $X$ of the parents
	of $y$, a set of variables $I$ fulfills the \emph{permutation-free
	instrumental set condition} if: (1) There are no treks from $I$ to $y$
	in the graph ${\cal G}_{\overline{X}}$ obtained by removing all arrows
	leaving $X$, and (2) All $t$-separators $(L,R)$ of $I$ and $X$ have
	size $\geq k$. 
\end{definition}

\begin{theorem}
\label{thm:graphicistrek}
The instrumental set criterion is equivalent to the permutation-free
instrumental set criterion.
\end{theorem}

\begin{proof}
 This is shown by adapting a closely related existing result \citep{ZanderL16}. 
 See Supplement for details.
\end{proof}

\begin{definition}[Algebraic Instrumental Sets (\citet{bollen1996alternative}, \citet{bollen2012instrumental})]
	\label{defn:algebraicis}
	Given a regression equation $y = B \cdot X + \epsilon$, where $X$ possibly
	correlates with $\epsilon$, a set of variables
	$I$ fulfills the \emph{algebraic instrumental set condition} if: (1) $I \ci \epsilon$,
	(2) $\textrm{rk}(\bm{\Sigma}[I,X]) = |X|$, and (3) $\textrm{rk}(\bm{\Sigma}[I]) = |I|$
\end{definition}

Having rephrased the instrumental set criterion 
without relying on permutations, we can now establish a
correspondence to the algebraic condition for instrumental variables -- which also
serves as an alternative correctness proof for Definition~\ref{defn:graphicis}
itself. The proof of the Theorem is included in the Supplementary Material.

\begin{theorem}
\label{thm:algebraictographical}
Given an SEM $(\bm{B},\bm{\Phi})$ with path diagram ${\cal G}=(V,A)$ 
and a variable $y \in V$, let
$X$ be a subset of the parents of $y$ in $\cal G$. Then
a set of variables $I \subseteq V$ fulfills the algebraic instrumental set 
condition with respect to the equation
\[
	y = B \cdot X + \epsilon ; \textit{ where } \epsilon = \sum_{p \in \textrm{Pa}_{\cal G}(y) \setminus X} p + \epsilon_y
\]
if and only if $I$ fulfills the instrumental set 
condition with respect to $X$ and $y$ in $\cal G$.
\end{theorem}

In the R package MIIVsem \citep{fisher2019miivsem} implementation of MIIV, all
parameters in an equation of an SEM are simultaneously identified by (1)
applying an L2O transformation to all the latent variables in this equation;
(2) identifying the composite error term of the resulting equation; and (3)
applying the algebraic instrumental set criterion based on the model matrices
initialized with arbitrary parameter values and derived total effect and
covariance matrices; see \citet{bollen2004automating} for details.
Theorem~\ref{thm:algebraictographical} implies that the MIIVsem approach is
generally equivalent to first applying the graphical L2O transform followed by
the instrumental set criterion (Definition~\ref{defn:graphicis}) using the set
of all observed parents of the dependent variable in the equation as $X$.

\section{\MakeUppercase{Examples}}
\label{sec:examples}

\begin{figure}[t]
	\begin{subfigure}[b]{0.5 \linewidth}
		\centering
		\includegraphics[page=7]{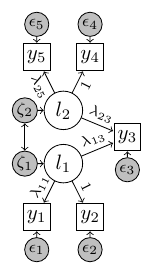}
		\caption{}
		\label{fig:transform_example}
	\end{subfigure}%
	\begin{subfigure}[b]{0.5 \linewidth}
		\centering
		\includegraphics[page=8]{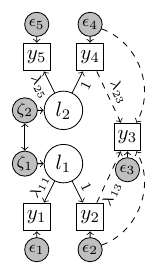}
		\caption{}
		\label{fig:transform_example_single}
	\end{subfigure}
\caption{(a) Example model following the structure of
	Figure~\ref{fig:example_sem} with explicit error terms. 
	(b) L2O transformation for the model in (a) for identifying both
	coefficients of the equation for $y_3$ simultaneously. We end up with
	the regression equation $y_3 \sim y_2 + y_4$ and can identify both
	coefficients using $y_1$ and $y_5$ as instrumental variables.}
\label{fig:examples1}
\end{figure}

Having shown that the algebraic instrumental set criterion is equivalent to the
graphical instrumental set criterion, we now show some examples of
identification using the proposed graphical approach and compare it to the MIIV 
approach implemented in MIIVsem \footnote{In some examples, a manual implementation
of the MIIV approach can permit estimation of models that are not covered by the
implementation in MIIVsem}.
First, we show an example of a full equation identification where we identify
all parameters of an equation altogether. Second, we show an example of partial
L2O transformation (as shown in Section~\ref{sec:partial_transform}) that
allows us to estimate a subset of the parameters of the equation. Third, we
show an example where the instrumental set criterion fails to identify any
parameters, but the conditional instrumental set criterion \citep{BritoP02} can
still identify some parameters. Finally, we show an example where the
parameters are inestimable even though the equation is identified.


\subsection{Identifying Whole Equations}
\label{sec:example_whole}
In this section, we show an example of identifying a whole equation using the
graphical criterion. Let us consider an SEM adapted from \citet{Shen2001}, as
shown in Figure~\ref{fig:examples1}a. We are interested in estimating the
equation $ y_3 \sim l_1 + l_2 $, i.e., parameters $ \lambda_{13} $ and $
\lambda_{23} $. Doing a graphical L2O transformation for both these parameters
together adds the edges $ y_2 \to y_3 $, $ y_4 \to y_3 $, $ \epsilon_2 \to y_3
$, and $ \epsilon_4 \to y_3 $, and removes the edges $ l_1 \to y_3 $, and $ l_2
\to y_3 $, resulting in the model shown in Figure~\ref{fig:examples1}b. Now, for
estimating $ \lambda_{13} $ and $ \lambda_{23} $ we can use the regression
equation  $ y_3 \sim y_2 + y_4 $, with $ y_1 $ and $ y_5 $ as the IVs. As $ y_1
$ and $ y_5 $ satisfy Definition~\ref{defn:graphicistrek}, both the parameters
are identified. Both of these parameters are also identifiable using MIIVsem.

\subsection{Identifying Partial Equations}
\begin{figure}[t]
	\centering
	\begin{subfigure}[b]{0.33 \linewidth}
		\centering
		\includegraphics[page=11]{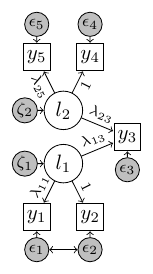}
		\caption{}
		\label{fig:example_orig}
	\end{subfigure}%
	\begin{subfigure}[b]{0.33 \linewidth}
		\centering
		\includegraphics[page=12]{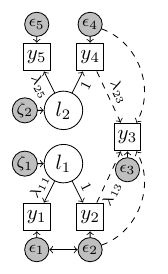}
		\caption{}
		\label{fig:example_non_corr}
	\end{subfigure}%
	\begin{subfigure}[b]{0.33 \linewidth}
		\centering
		\includegraphics[page=13]{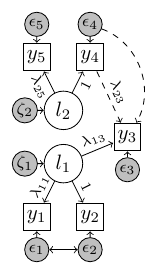}
		\caption{}
		\label{fig:transform_non_corr}
	\end{subfigure}
	\caption{(a) Adapted SEM from \citet{Shen2001}; modified by making $ l_1
		 $ and $ l_2 $ uncorrelated and $ \epsilon_1 $ and $ \epsilon_2
		 $ correlated. (b) Transformed model for estimating $ y_3 \sim
		 l_1 + l_2 $. The equation is not identified as $ y_5 $ is the
		 only IV. (c) With partial L2O transformation, $ \lambda_{23} $
		 can be estimated using $ y_5 $ as the IV.}
\label{fig:examples3}
\end{figure}

For this section, we consider a slightly modified version of the model in the
previous section. We have added a correlation between $ \epsilon_1 $ and
$\epsilon_2$, and have allowed the latent variables, $l_1$ and $l_2$ to be
uncorrelated, as shown in Figure~{\ref{fig:example_orig}}. The equation $ y_3
\sim l_1 + l_2 $ is not identified in this case, as $ y_5 $ is the only
available IV (Figure~\ref{fig:example_non_corr}). However, using the partial
graphical transformation for $ l_2 $ while treating $ l_1 $ as an error term
(Figure~{\ref{fig:transform_non_corr}}), the parameter $ \lambda_{23} $ can be
identified by using $y_5$ as the IV. As the R package MIIVsem always tries to
identify full equations, it is not able to identify either of the parameters in
this case -- although this would be easily doable when applying the MIIV
approach manually.



\subsection{Identification Based on Conditional IVs}
\begin{figure}[t]
	\centering
	\begin{subfigure}[b]{0.5 \linewidth}
		\centering
		\includegraphics[page=14]{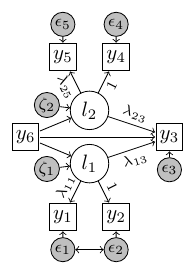}
		\caption{}
		\label{fig:example_conditional_iv}
	\end{subfigure}%
	\begin{subfigure}[b]{0.5 \linewidth}
		\centering
		\includegraphics[page=15]{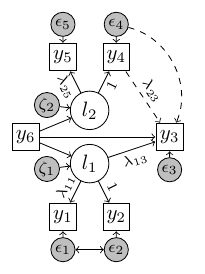}
		\caption{}
		\label{fig:transform_conditional_iv}
	\end{subfigure}
	\caption{(a) Modified version of the Figure~\ref{fig:example_orig} model, where $l_1$ and $l_2$
		share an observed cause $y_6$.
		$\lambda_{13} $ and $ \lambda_{23} $ are still not
		simultaneously identified as no IVs are available. (b) Even
		with partial transformation, $ \lambda_{23} $ is no longer
		identified as $ y_5 $ is not an IV because of the open paths $
		y_5 \gets l_2 \gets y_6 \to y_3 $ and $ y_5 \gets l_2 \gets y_6
		\to l_1 \to y_3 $. However, using the conditional instrumental
		set criterion, we can identify $\lambda_{23}$ by using $y_5$ as
		a \emph{conditional IV} for the equation $y_3 \sim y_4$, as
		conditioning on $ y_6 $ blocks the open paths.}
\label{fig:examples4}
\end{figure}

So far, we have only considered the instrumental set criterion, but many
other identification criteria have been proposed for DAGs. For example, we can
generalize the instrumental set criteria to hold \emph{conditionally}
on some set of observed variables \citep{BritoP02}. There can be cases when
conditioning on certain variables allows us to use
conditional IVs. This scenario might not occur when we have a standard latent
and measurement level of variables, but might arise in specific cases; for
example, when there are exogenous covariates that can be measured without error
(such as the year in longitudinal studies), or interventional variables in
experimental settings (such as complete factorial designs) which are
uncorrelated, and observed exogenous by definition.
Figure~{\ref{fig:example_conditional_iv}} shows a hypothetical example in which
the latent variables $ l_1 $ and $ l_2 $ are only correlated through a common
cause $ y_6 $, which could, for instance, represent an experimental
intervention. Similar to the previous example, a full identification for $ y_3
\sim l_1 + l_2 + y_6 $ still does not work. Further, because of the added correlation
between $ l_1 $ and $ l_2 $, partial identification is not possible either. The
added correlation between $l_1$ and $l_2$ opens a path from $y_5$ to $y_3$,
resulting in $y_5$ no longer being an IV for $ y_3 \sim y_4 $. However, the
conditional instrumental set criterion \citep{BritoP02} can be used here to
show that the parameter $ \lambda_{23} $ is identifiable by conditioning on
$y_6$ in both stages of the IV regression. In graphical terms, we say that
conditioning on $y_6$ $d$-separates the path between $l_1$ and
$l_2$ (Figure~\ref{fig:transform_conditional_iv}),  which means that we end up
in a similar situation as in Figure~\ref{fig:transform_non_corr}. We can
therefore use $y_5$ as an IV for the equation $y_3 \sim y_4$ once we condition
on $y_6$. As the MIIV approach does not consider conditional IVs, it is not able
to identify either of the parameters.

\subsection{Inestimable Parameters in Identified Equations}

\begin{figure}[t]
	\centering
	\begin{subfigure}[b]{0.33 \linewidth}
		\centering
		\includegraphics[page=16]{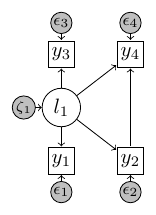}
		\caption{}
		\label{fig:example_original_model}
	\end{subfigure}%
	\begin{subfigure}[b]{0.33 \linewidth}
		\centering
		\includegraphics[page=17]{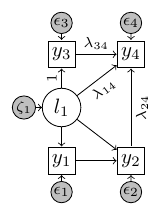}
		\caption{}
		\label{fig:example_incorrect_estimate}
	\end{subfigure}%
	\begin{subfigure}[b]{0.33 \linewidth}
		\centering
		\includegraphics[page=18]{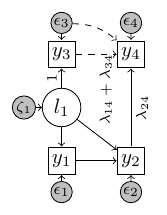}
		\caption{}
		\label{fig:transform_incorrect_estimate}
	\end{subfigure}%
	\caption{(a) An example model from \citet{griliches1977estimating}
	about the economic effects of schooling. The model has $ 1 $ latent
	variable $x_1$ (Ability) with $ 4 $ observed variables $ y_1 $ (IQ), $
	y_2 $ (Schooling), $ y_3 $ (knowing how the world works), and $ y_4 $
	(Income). (b) A slightly modified version of the model in
	Figure~\ref{fig:example_original_model} where we add two new edges $ y_1
	\rightarrow y_2 $ and $ y_3 \rightarrow y_4 $.  (c) L2O transformed
	model for the equation of $ y_4 $. The transformed regression equation
	for $ y_4 $ is: $ y_4 \sim y_3 + y_2 $ but because of the
	transformation, the coefficient of $ y_3 $ has changed
	to $ \lambda_{14} + \lambda_{34} $. Because of this changed
	coefficient, even though the equation is identified, it is not possible
	to estimate either $ \lambda_{14} $ or $ \lambda_{34} $ individually.
	}
\label{fig:examples5}
\end{figure}

In the previous examples, the L2O transformation creates a new edge in the
model between two observed variables that has the same path coefficient that we
are interested in estimating. But if the L2O transformation adds a new edge
where one already exists, the new path coefficient becomes the sum of the
existing coefficient and our coefficient of interest. In such cases, certain
parameters can be inestimable even if the transformed equation is
identified according to the identification criteria.

In Figure~\ref{fig:example_original_model}, we have taken a model about the
economic effects of schooling from \citet{griliches1977estimating}. 
All parameters in the equation of $y_4$ are identifiable by using
$y_1$ and $y_2$ as the IVs. However, we get an interesting case if we
add two new edges $ y_1 \rightarrow y_3 $ and $ y_2 \rightarrow y_4 $
(Figure~\ref{fig:example_incorrect_estimate}): The L2O transformation for the
equation of $ y_4 $ adds the edges $ \epsilon_3 \rightarrow y_4 $ and $ y_3
\rightarrow y_4 $, as shown in Figure~\ref{fig:transform_incorrect_estimate}.
But since the original model already has the edge $ y_3 \rightarrow y_4 $, the
new coefficient for this edge becomes $ \lambda_{14} + \lambda_{34} $. The
regression equation for $y_4$ is still: $ y_4 \sim y_3 + y_2 $, and it is
identified according to the instrumental set criterion as $ y_2 $ and $ y_1 $
are the IVs for the equation. But if we estimate the parameters, we will obtain
values for $ \lambda_{24} $ and $ \lambda_{14} + \lambda_{34} $. Therefore, $
\lambda_{24} $ remains identifiable in this more general case, but $
\lambda_{14} $ and $\lambda_{34}$ are individually not identified. The graphical
L2O approach allows us to easily visualize such cases after transformation.


\section{\MakeUppercase{Discussion}}

In this paper, we showed the latent-to-observed (L2O) transformation on the RAM
notation and how to use it for partial equation identification.  We then
gave an equivalent graphical L2O transformation which allowed us to apply
graphical identification criteria developed in the DAG literature to latent
variable parameters in SEMs. Combining this graphical L2O transformation with
the graphical criteria for parameter identification, we arrived at a generic
approach for parameter identification in SEMs. Specifically, we showed that the
instrumental set criterion combined with the graphical L2O transformation is
equivalent to the MIIV approach. Therefore, the graphical transformation can be
used as an explicit visualization of the L2O transformation or as an
alternative way to implement the MIIV approach in computer programs.  To
illustrate this, we have implemented the MIIV approach in the graphical-based R
package \emph{dagitty} \citep{Textor2017} and the Python package \emph{pgmpy} \citep{ankan2015pgmpy}.

Our equivalence proof allows users to combine results from two largely disconnected lines of work.
By combining the graphical L2O transform with other identification criteria, 
we obtain novel identification strategies for LVSEMs, as we have illustrated 
using the conditional instrumental set criterion. Other
promising candidates would be auxiliary
variables \citep{ChenKB17} and instrumental cutsets \citep{KumorCB19}. 
Conversely, the SEM literature is more developed than the graphical literature
when it comes to non-Gaussian models. For example,
MIIV with two-stages least squares estimation is asymptotically distribution-free
\citep{bollen1996alternative},
and our results imply that normality is not required for applying the
instrumental set criterion.

\bibliography{bibliography.bib}

\appendix

\section{L2O Transformation for LISREL Models}

\label{sec:lisrel}

In this section, we show the LISREL notation of SEMs and L2O transformation as presented in \citet{bollen1996alternative}.

\subsection{LISREL Notation}

The LISREL notation of SEMs was first introduced in the LISREL (LInear
Structural Relation) software \citep{joreskog1993lisrel}. This notation 
is based on the assumption that the models have an underlying latent structure
and that the only observed variables are those that act as the measurement variables for these
latents.  This assumption allows us to split the set of model equations into
two subsets representing: the latent model and the measurement model as follows:

\begin{equation}
\label{eq:lisrel}
\begin{split}
	\textit{Latent Model:} & \\
		\eta &= \bm{B} \eta  + \bm{\Gamma} \xi + \zeta \\
	\textit{Measurement Model:} \\
		Y &= \bm{\Lambda}_Y \eta + \epsilon \\
		X &= \bm{\Lambda}_X \xi + \delta
\end{split}
\end{equation}

Here, $ \eta $ ($ \xi $) is the sets of endogenous (exogenous) latent
variables, and $ Y $ ($ X $) is the set of observed measurement
variables for $ \eta $ ($ \xi $). $ \bm{B} $, $ \bm{\Gamma} $, $
\bm{\Lambda}_Y $, and $ \bm{\Lambda}_X $ are the parameter matrices specifying
the path coefficients in the model. $ \zeta $, $ \epsilon $, $ \delta $ are the
error vectors with the covariance matrix $ \bm{\Phi}_{\zeta} $, $
\bm{\Phi}_{\epsilon} $, and $ \bm{\Phi}_{\delta} $ respectively. An example of
an SEM in LISREL notation along with its path model representation is shown in
Figure~\ref{fig:small_sem}.

From the model equations, it appears that many possible variable relations
cannot be specified directly. For example, it is not clear how to specify
direct relations between two observed variables, between an observed and latent
variable, or an error correlation between $ \zeta $ and $ \epsilon $ terms. But
these relations can be modelled in the LISREL notation by making some simple
modifications to the model \citep{bollen1989}. For example, for adding a direct
causal relation between two observed variables, we can instead use two latent
variables (with the same causal direction and path coefficient), and add the
actual observed variables as single measurement variables for each latent,
fixing the measurement errors for these relations to 0. This modification
transforms the model into having a latent and measurement levels which can be
represented in the LISREL notation.

\begin{figure}[h]
     \begin{subfigure}[b]{0.5\textwidth}
    	\centering
		\includegraphics[page=3]{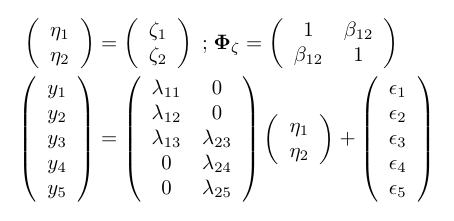}
        \caption{}
	\label{fig:small_sem_eq}
    \end{subfigure}
    \begin{subfigure}[b]{0.5\textwidth}
    	\centering
    	\includegraphics[page=2]{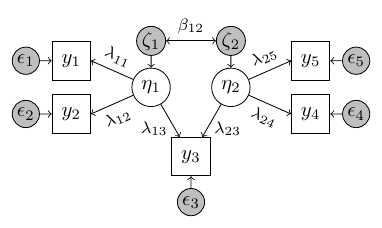}
	\caption{}
	\label{fig:small_sem_diagram}
    \end{subfigure}
    \caption{Example of an SEM in LISREL notation.
	The parameters $ \beta_{ij} $ and $ \lambda_{ij}
	$ are the so-called \emph{path coefficients} on latent-latent and
	latent-observed arrows, respectively.  The bi-directed arrow represents
	a correlation between the error terms that is allowed to be nonzero,
	reflecting our belief that $ \eta_1 $ and $ \eta_2 $ may be correlated.
	(a) Model in equation form; (b) path diagram.}
    \label{fig:small_sem}
\end{figure}

\subsection{Algebraic L2O Transformation for LISREL Models}

We now introduce the L2O transformation as shown in
\citet{bollen1996alternative}.  Let us assume that every latent variable in the
model has a unique scaling indicator that is not an indicator of any other
latent variable. Then we can replace this latent variable by the difference of
its scaling indicator and the scaling indicator's error term. For instance,
applying this to the latent variable $ \eta_1 $ with $ y_1 $ as its scaling
indicator in the model in Figure{~\ref{fig:small_sem}} we would get

\[
\eta_1 = y_1 - \epsilon_1
\]

Applying this L2O transformation to all latent variables in the
Equation~\ref{eq:lisrel} simultaneously, we get:

\begin{equation}
	\label{eq:lisrel_obs_single_comb}
	\begin{split}
	Y_1 &= \bm{B} Y_1 + \bm{\Gamma} X_1 + \epsilon_1 - \bm{B} \epsilon_1 - \bm{\Gamma} \delta_1 + \zeta \\
	Y_2 &= \bm{\Lambda}_{Y_2} Y_1 - \bm{\Lambda}_{Y_2} \epsilon_1 + \epsilon_2 \\
	X_2 &= \bm{\Lambda}_{X_2} X_1 - \bm{\Lambda}_{X_2} \delta_1 + \delta_2
	\end{split}
\end{equation}

where $ X_1 $ and $ Y_1 $ are the scaling indicators for $ \eta $ and $ \xi $
respectively. $ X_2 $ and $ Y_2 $ are the remaining observed variables, $ X_2 = X
\setminus X_1 $ and $ Y_2 = Y \setminus Y_1 $. $ \Lambda_{X_2} $ and $
\Lambda_{Y_2} $ are submatrices of $ \Lambda_X $ and $ \Lambda_Y $ with only
rows corresponding to $ X_2 $ and $ Y_2 $. Similarly, the error terms $
\epsilon_i $ and $ \delta_i $ correspond to $ Y_i $ and $ X_i $ respectively.

As a result of applying L2O transformation to all the latent variables,
Equation~\ref{eq:lisrel_obs_single_comb} now only contains observed variables
and each of the individual equations now resembles a standard regression
equation.  However, by construction, the error terms of these equations can be
correlated with the covariates. This means that applying a standard
least-squares estimator will provide biased estimates for the model parameters.
To get unbiased estimates we can instead use an Instrumental Variable (IV)
based estimator like 2-SLS (Two-Stage Least Squares)
\citep{bollen1996alternative}. 

\section{Implied Covariance Matrix and Trek Rule}
In this section, we show how the implied covariance matrix that we use in the
paper is related to the model parameters. We show this both in algebraic and
graphical terms.

SEMs can be rewritten to a canonical form that does not require correlations
between error terms. For this, we introduce a new variable $x_{ij}$ for each
upper triangular nonzero entry $\bm{\Phi}[i,j]$ and set $\beta_{x_{ij} \to
x_i}=\beta_{x_{ij} \to x_j}=1$ and $\phi_{\epsilon_{ij}}=\bm{\Phi}[i,j]$.  The
\emph{implied covariance matrix} $\bm{\Sigma}$ of a canonical SEM is given by :
\[
\bm{\Sigma}=\bm{B}^{-T}\bm{\Phi}\bm{B}^{-1}
\]

The \emph{trek rule} \citep{sullivant2010trek} allows us to express covariances
in SEMs in graphical terms:
\begin{equation}
\label{eqn:trekrule}
\bm{\Sigma}[i,j] = \sum_{\text{ treks }\pi\text{ from } x_i  \text{ to } x_j } \; 
\phi_{\pi^\leftrightarrow}\; \prod_{\text{arrows }k\to l \text{ on trek } \pi} \beta_{k \to l} 
\end{equation}
Here it is important to note that treks can contain the same nodes
twice; this is required for the trek rule to work. 

\section{Proofs}
In this section, we give proofs for the theorems in the paper.

\subsection{Proof of Lemma~\ref{lemma:scaling}}
\begin{proof}
We assume that the SEM has been transformed to canonical form with
no bi-directed arrows. The covariance of two variables variables $u,v \neq x_i$,
is given by the trek rule (Equation~\ref{eqn:trekrule}).
Therefore $\bm{\Sigma}[u,v]$ will be the same in $\cal G$ and $\cal G'$ 
if no treks from $u$ to $v$ pass through $x_i$. 
Otherwise, let $\pi$ be a trek in ${\cal G}'$ from $u$ to $v$ that 
includes $x_i$. There are two cases. (1) $x_i$ is the top of $\pi$. 
Then $\pi$ differs between $\cal G$ and $\cal G'$ in the
 sub-trek $x_j \gets x_i \to x_k$ with path coefficient 
product in $\cal G'$ of $\alpha \beta_{x_i \to x_j} \cdot \alpha^{-2} \phi_{x_i} \cdot \alpha \beta_{x_i \to x_k} = \beta_{x_i \to x_j} \cdot \phi_{x_i} \cdot \beta_{x_i \to x_k}$ which is the same as in $\cal G$.
(2) $\pi$ has a sub-trek $x_j \to x_i \to x_k$ with path 
coefficient product $\alpha \beta_{x_j \to x_i} \cdot  \alpha^{-1} \beta_{x_i \to x_k}$,
which again is the same as in $\cal G$. 
\end{proof}

\subsection{Proof of Theorem~\ref{thm:graphicistrek}}
\begin{proof}

The first conditions in both criteria are equal, so we prove the equivalence between 
the second conditions.  

$\Rightarrow$: Suppose the condition (2) of the instrumental set criterion
(Definition~\ref{defn:graphicis}) is satisfied. We need to show that no two paths $\pi_i$ and $\pi_j$ can be $t$-separated
by one variable. Indeed, $\pi_i$ and $\pi_j$ can only be $t$-separated by one 
variable if they have a same-sided intersection $v$. But this would contradict
condition (2) of Definition~\ref{defn:graphicis}. Consequently, we need $k$ variables
to $t$-separate all paths, and since these paths are a subset of the paths
from $\mathbf{X}$ to $\mathbf{Z}$, we cannot $t$-separate these variable sets with
fewer paths either.

$\Leftarrow$: Suppose condition (2) of the trek-based instrumental set criterion
(Definition~\ref{defn:graphicistrek})
is satisfied. Then there must exist sets of $k$ treks 
from $I$ to $X$; let
 $\pi_1, \ldots, \pi_k$ be one such set of treks with minimal total length.
No $\pi_i$ intersects any $\pi_j, i \neq j,$ on the same side, otherwise we could separate
all paths with $k-1$ variables. Therefore, all intersections between the $\pi_i$
are opposite sided. Define an ordering $\preceq$ on the $\pi_i$ as follows: 
$\pi_i \preceq \pi_j$ if $\pi_i$ and $\pi_j$ intersect at a variable $k$, which 
is on the left side of $\pi_i$ and on the right side of $\pi_j$ (note that $k$ cannot be the 
top of $\pi_i$ or $\pi_j$).

Suppose that the $\pi_i$ contain a cycle of length $l$ with respect to $\preceq$, that is, 
$\pi_{i_1} \preceq \ldots \preceq \pi_{i_l} \preceq \pi_{i_{1}}$. Then we 
can combine a prefix of each trek in the cycle with a suffix of the next trek  
to create $l$ other treks between 
the same variables that cannot be separated by fewer than $l$ variables, since they 
do not have same-sided intersections. But these
new paths would be shorter than the ones on the cycle, a contradiction (see
Figure~1 for an example). 
Hence such a cycle
cannot exist, and the paths can be linearly ordered with respect to $\preceq$. Any
such ordering fulfills requirement (b) of condition (2) in Definition~\ref{defn:graphicis}.

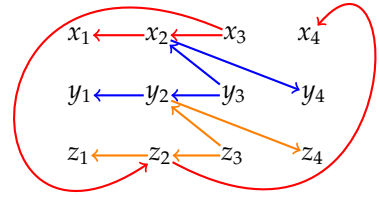
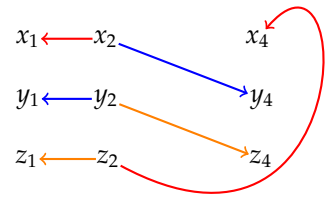
\begin{figure}[t]
\centering
\begin{subfigure}[b]{\linewidth}
\begin{tikzpicture}[xscale=1.2,inner sep=.05pt,thick]
\node (v0) at (-2.16,1.4) {$x_1$};
\node (v1) at (-1.30,1.4) {$x_2$};
\node (v2) at (-0.435,1.4) {$x_3$};
\node (v3) at (0.385,1.4) {$x_4$};
\node (v4) at (-2.16,0.6) {$y_1$};
\node (v5) at (-1.30,0.6) {$y_2$};
\node (v6) at (-0.457,0.6) {$y_3$};
\node (v7) at (0.428,0.6) {$y_4$};
\node (v8) at (-2.17,-0.2) {$z_1$};
\node (v9) at (-1.27,-0.2) {$z_2$};
\node (v10) at (-0.468,-0.2) {$z_3$};
\node (v11) at (0.423,-0.2) {$z_4$};
\draw [->,red] (v1) edge (v0);
\draw [->,red] (v2) edge (v1);
\draw [->,blue] (v5) edge (v4);
\draw [->,blue] (v6) edge (v5);
\draw [->,blue] (v6) edge (v1);
\draw [->,blue] (v1) edge (v7);
\draw [->,orange] (v9) edge (v8);
\draw [->,orange] (v10) edge (v9);
\draw [->,orange] (v10) edge (v5);
\draw [->,orange] (v5) edge (v11);
\draw [->,red] (v2) .. controls (-3.5,3) and (-3.5,-2) .. (v9);
\draw [->,red] (v9) .. controls (1.5,-2) and (1.5,3) .. (v3);
\end{tikzpicture} \\
\caption{}
\end{subfigure}
\begin{subfigure}[b]{\linewidth}
\centering
\begin{tikzpicture}[xscale=1.2,inner sep=.05pt,thick]
\node (v0) at (-2.16,1.4) {$x_1$};
\node (v1) at (-1.30,1.4) {$x_2$};

\node (v3) at (0.385,1.4) {$x_4$};
\node (v4) at (-2.16,0.6) {$y_1$};
\node (v5) at (-1.30,0.6) {$y_2$};

\node (v7) at (0.428,0.6) {$y_4$};
\node (v8) at (-2.17,-0.2) {$z_1$};
\node (v9) at (-1.27,-0.2) {$z_2$};

\node (v11) at (0.423,-0.2) {$z_4$};
\draw [->,red] (v1) edge (v0);
\draw [->,blue] (v5) edge (v4);
\draw [->,blue] (v1) edge (v7);
\draw [->,orange] (v9) edge (v8);
\draw [->,orange] (v5) edge (v11);
\draw [->,red] (v9) .. controls (1.5,-2) and (1.5,3) .. (v3);
\end{tikzpicture} \\
\caption{}
\end{subfigure}
\caption{(a) Paths $\pi_1 = x_1 \gets x_2 \gets x_3 \to z_2 \to x_4$,
$\pi_2 = y_1 \gets y_2 \gets y_3 \to x_2 \to y_4$ and 
$\pi_3 = z_1 \gets z_2 \gets z_3 \to y_2 \to z_4$ where 
$\pi_1 \preceq \pi_2 \preceq \pi_3 \preceq \pi_1$. (b)
By rearranging segments of these paths, we obtain shorter paths 
$\pi_1' = x_1 \gets x_2 \to y_4$,
$\pi_2' = y_1 \gets y_2 \to z_4$ and 
$\pi_3' = z_1 \gets z_2 \to x_4$ between the same variables that 
intersect less then the original paths and therefore cannot be 
$t$-separated by fewer variables.
}
\end{figure}

Now assume requirement (a) is violated, that is, there exist paths $\pi_i$ from $Z_i$ to $X_i$
and $\pi_j$ from $Z_j$ to $X_j$ such that $i < j$ and $Z_j$ also occurs on $\pi_i$. Then 
$Z_j$ cannot be on the left side of $\pi_i$ because then $Z_j$ would be a same-sided intersection
of $\pi_i$ and $\pi_j$. But if $Z_j$ is on the right side of $\pi_i$, then $\pi_j \preceq \pi_i$,
a contradiction. So requirement (a) must be fulfilled as well.

\end{proof}

\subsection{Proof of Theorem~\ref{thm:algebraictographical}}
\begin{proof}
Condition 3 of the algebraic instrumental set criterion holds by definition as 
we require the covariance matrix $\bm{\Phi}$ to be positive definite (in other
words, we do not allow deterministic relations). It remains to be shown 
that the first two conditions of both criteria are equivalent. For condition (1),
assume that  some $i \in I$ is not independent of some parent $p$ of $y$, then there must be 
a trek from $i$ to $p$ which can be extended to $y$. Conversely, assume 
that there is a trek $\pi$ from $i$ to $y$ in ${\cal G}_{\overline{X}}$.
Then $\pi$ ends with an arrow $p \to y$ where $p \notin X$, so $I$ is not 
independent of the composite error term $\epsilon$. For condition (2), the equivalence
follows directly from Theorem~\ref{thm:treksep}.
\end{proof}

\end{document}